\documentclass[12pt, draftclsnofoot, onecolumn]{IEEEtran}

\usepackage[utf8]{inputenc}
\IEEEoverridecommandlockouts
\usepackage{cite}
\usepackage{xcolor}
\usepackage{amsmath,amssymb,amsfonts}
\usepackage{algorithm}
\usepackage{algorithmic}
\usepackage{etoolbox}  
\usepackage{bbm}
\usepackage{graphicx}
\usepackage{amsmath}
\usepackage{amssymb}
\usepackage{amsthm}
\usepackage{bm}
\usepackage{gensymb}
\usepackage{subfig}
\usepackage{booktabs}
\usepackage{tikz}
\usepackage{pgfplots}
\usepackage{latexsym}
\pgfplotsset{compat=1.5}
\usetikzlibrary{automata,arrows,positioning,calc,shapes, intersections}
\usepgfplotslibrary{colorbrewer}
\DeclareMathOperator*{\argmin}{arg\,min\,}
\DeclareMathOperator*{\argmax}{arg\,max\,}

\newcommand{\set}[1]{\mathcal{#1}}
\newtheorem{definition}{Definition}

\newtheorem{theorem}{Theorem}

\makeatletter
\patchcmd{\algorithmic}{\addtolength{\ALC@tlm}{\leftmargin} }{\addtolength{\ALC@tlm}{\leftmargin}}{}{}
\makeatother

\makeatletter
\newcommand\fs@betterruled{%
	\def\@fs@cfont{\bfseries}\let\@fs@capt\floatc@ruled
	\def\@fs@pre{\vspace*{5pt}\hrule height.8pt depth0pt \kern2pt}%
	\def\@fs@post{\kern2pt\hrule\relax}%
	\def\@fs@mid{\kern2pt\hrule\kern2pt}%
	\let\@fs@iftopcapt\iftrue}
\floatstyle{betterruled}
\restylefloat{algorithm}
\makeatother

\definecolor{green}{RGB}{27,158,119}
\definecolor{BuGn}{RGB}{28,144,153}
\newcommand{\ilm}{\textcolor{black}}

\usepackage{filecontents}
\usepackage{graphicx}
\usepackage{textcomp}
\usepackage{xcolor}
\def\BibTeX{{\rm B\kern-.05em{\sc i\kern-.025em b}\kern-.08em
		T\kern-.1667em\lower.7ex\hbox{E}\kern-.125emX}}

\begin{document}
\title{Inter-Plane Inter-Satellite
Connectivity in Dense LEO Constellations}
\author{Israel~Leyva-Mayorga, \IEEEmembership{Member,~IEEE}, Beatriz~Soret, \IEEEmembership{Member,~IEEE}, and Petar~Popovski, \IEEEmembership{Fellow,~IEEE}
\thanks{The authors are with the Department of Electronic
Systems, Aalborg University, 9220 Aalborg, Denmark (e-mail:\{ilm, bsa, petarp\}@es.aau.dk)}}
\maketitle
\begin{abstract}
With numerous ongoing deployments owned by private companies and startups, dense satellite constellations deployed in low Earth orbit (LEO) will play a major role in the near future of wireless communications. In addition, the 3rd Generation Partnership Project (3GPP) has ongoing efforts to integrate satellites into 5G and beyond-5G networks. Nevertheless, numerous challenges must be overcome to fully exploit the connectivity capabilities of satellite constellations. These challenges are mainly a consequence of the low capabilities of individual small satellites, along with their high orbital speeds and small coverage due to the low altitude of deployment. In particular, inter-plane inter-satellite links (ISLs), which connect satellites from different orbital planes, are greatly dynamic and may be considerably affected by the Doppler shift. In this paper, we present a framework and the corresponding algorithms for the dynamic establishment of the inter-plane ISLs in LEO constellations. Our results show that the proposed algorithms increase the sum of rates in the constellation 1) by up to $115$\% with respect to the state-of-the-art benchmark schemes in an interference-free environment and 2) by up to $71$\% when compared to random resource allocation in a worst-case scenario for interference.
\end{abstract}


\IEEEpeerreviewmaketitle

\section{Introduction}
\label{sec:intro}

There is an unprecedented interest from the industry and international agencies on dense satellite constellations deployed in low Earth orbit (LEO). Due to their relatively low altitude of deployment, between $500$ and $2000$~km over the Earth's surface~\cite{3GPPTR38.821, TR22.822, Qu2017}, LEO constellations are able to provide global coverage and 
reduced propagation delays in the ground-to-satellite links (GSLs) when compared to higher orbits. This combination of characteristics makes them an appealing option to support two of the three main use cases for 5G: massive machine-type communications (mMTC) and enhanced mobile broadband (eMBB). In addition, LEO constellations can be used for ultra-reliable communications (URC) with high
availability (up to $99.99$\% for LEO constellations~\cite{3GPPTR38.821}) and reliability in combination with 
slightly relaxed latency requirements, in the order of a few tens of milliseconds~\cite{3GPPTR38.821,TR38.913,Giambene2018}. Hence, LEO constellations are envisioned to be integrated into 5G and beyond-5G wireless networks with the purpose to dramatically extend cellular coverage, serve as a global backbone, and offload the cellular base stations in problematic hot spots~\cite{TR22.822, TR38.913, 3GPPTR38.811, 3GPPTR38.821}. 

The position of LEO satellites with respect to the ground is not fixed. Instead, the orbital velocities of LEO satellites \ilm{are up to $7.6\,\mathrm{km}/\mathrm{s}$}; much greater than those of satellites in higher orbits. Therefore, LEO constellations are typically organized in groups of satellites that orbit the Earth following the same trajectory called \emph{orbital planes}. Moreover, the low altitude of deployment imposes an important limitation on the ground coverage of an individual LEO satellite\ilm{\cite{TR22.822}}. 

Due to the characteristics described above, dense LEO deployments are needed to provide global and continuous coverage\ilm{\cite{TR22.822}}. In combination with the necessity of reducing the overall cost of deployment of the constellation, this fosters the use of physically small satellites with low individual computing and connectivity capabilities. In addition, the high orbital velocities create frequent changes in the satellite network topology and complicate the communication between satellites in different orbital planes. 

 Communication between satellites takes place through inter-satellite links (ISLs); these are illustrated along with the satellite axes in Fig.~\ref{fig:constellation} for a typical Walker star constellation~\cite{Walker1984}. 
 Communication between satellites in the same orbital plane occurs through the intra-plane ISLs, using the antennas located at both sides of the roll axis. Intra-plane ISLs are rather stable due to the nearly constant distance between neighboring satellites in a same orbital plane, called \emph{intra-plane distance}. On the other hand, communication between satellites in different orbital planes occurs through the inter-plane ISLs, using the antennas located at both sides of the pitch axis. Contrary to intra-plane ISLs, inter-plane ISLs are highly dynamic due to the different velocity vectors of the satellites\ilm{~\cite{Su2019}}. 
\begin{figure*}[t]
    \centering
    \subfloat[]{\includegraphics{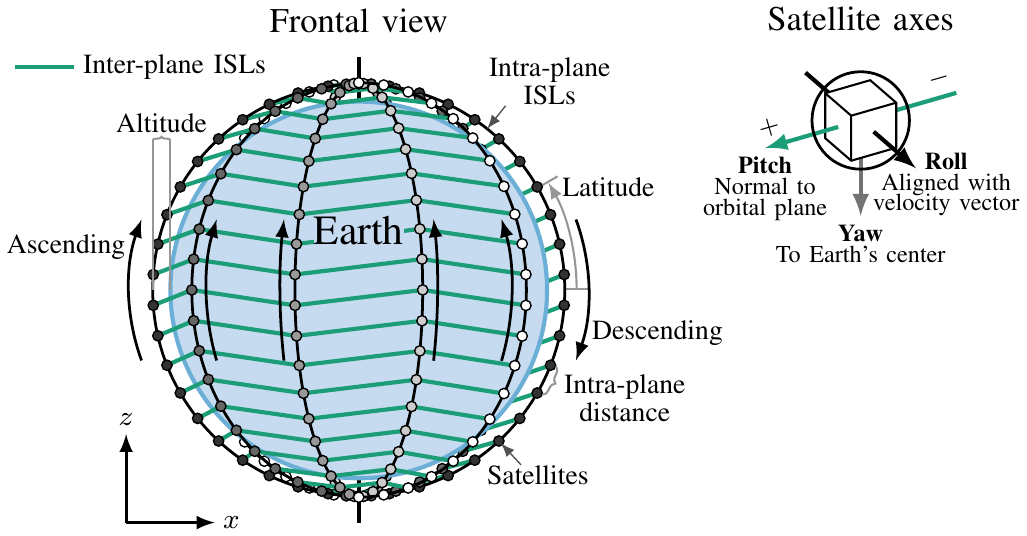}}
    \subfloat[]{\includegraphics{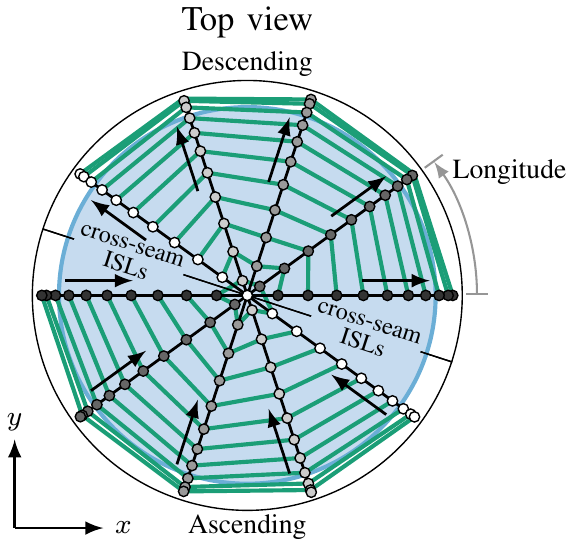}}
    \caption{(a) Frontal view and satellite axes and (b) top view of a Walker star constellation with $200$ satellites and $5$ orbital planes, deployed between $600$ and $640$~km above the Earth's surface.
    } 
    \label{fig:constellation} 
\end{figure*}

 In Fig.~\ref{fig:constellation}, the white and black satellites are orbiting the Earth in opposite directions, which results in large relative velocities between them. The ISLs between these orbital planes are known as \emph{cross-seam ISLs}, where the Doppler effect is considerably large. As a consequence, it is common to find constellations where the cross-seam ISLs are not implemented, as illustrated in Fig.~\ref{fig:constellation}~\cite{Liu2017, KeplerIEEESpectrum}. 

 \ilm{During the past decades, the connectivity aspects of satellite constellations have been widely investigated~\cite{Di2019, Kak2019, Ekici2001, Liu2017, Jiang2020,Hu2020}. However, only a few studies focus on inter-plane inter-satellite communication, even though it is essential to fully unleash the potential of LEO satellite constellations and facilitate their successful integration with 5G.} 
 
 Furthermore, most of the theoretical research on inter-satellite communication and routing considers a perfectly symmetric constellation~\cite{Kak2019, Ekici2001, Liu2017}, where orbital planes are deployed at the same altitude and at evenly-spaced longitudes (see Fig.~\ref{fig:constellation}). Naturally, a perfect symmetry greatly simplifies the ISL communication. For example, Kak et al.~\cite{Kak2019} focused on the design of a fully symmetrical constellation to maximize the coverage and throughput while minimizing the cost of deployment. On the other hand, Ekici \emph{et al.} proposed a routing algorithm for a fully symmetric Walker star constellation that exploits the horizontal alignment of the satellites to form \emph{rings}~\cite{Ekici2001}. 
 
 Nevertheless, the orbital planes of commercial LEO constellations are commonly deployed at slightly different altitudes. These differences, known as orbital separation, greatly minimize the risk of collisions between the satellites at the crossing points of the orbital planes~\cite{Lewis2019, OneWeb_brochure} (located at the poles in Fig.~\ref{fig:constellation}). If these orbital separations are not introduced, active station keeping may be necessary, which has a great cost in terms of propellant usage~\cite{Singh2020}. On the downside, orbital separations lead to slightly different orbital periods at the orbital planes, which lead to a dynamic network topology and make it impossible to use fixed tables to establish the inter-plane ISLs. Therefore, the inter-plane ISLs have to be established on-the-fly.
 
 In this paper, we formulate the establishment of unicast inter-plane ISLs in dense LEO constellations as a weighted dynamic matching problem and propose a framework to maximize the sum of rates selected for communication. Our framework, illustrated in Fig.~\ref{fig:examples}, encompasses two phases. In the first phase, the satellite pairs are selected from a set of feasible pairs $\set{E}$--the input--to create the satellite matching $\set{M}\subset \set{E}$. The goal is to maximize the sum of rates in $\set{M}$ assuming an interference-free environment. Next, in the second phase, \ilm{if the ISLs are affected by interference,} orthogonal wireless resources in the set $\set{K}$ are allocated to the selected satellite pairs $\set{M}$ to achieve an efficient resource sharing. Hence, set $\set{A}$ contains the resources allocated for communication between satellite pairs. To establish the inter-plane ISLs, $\set{A}$ must be informed to all satellites in the constellation and these must complete the necessary handshakes. The design and analysis of the routing and link establishment protocols are out of the scope of the paper.
 
 \begin{figure}
    \centering
    \includegraphics{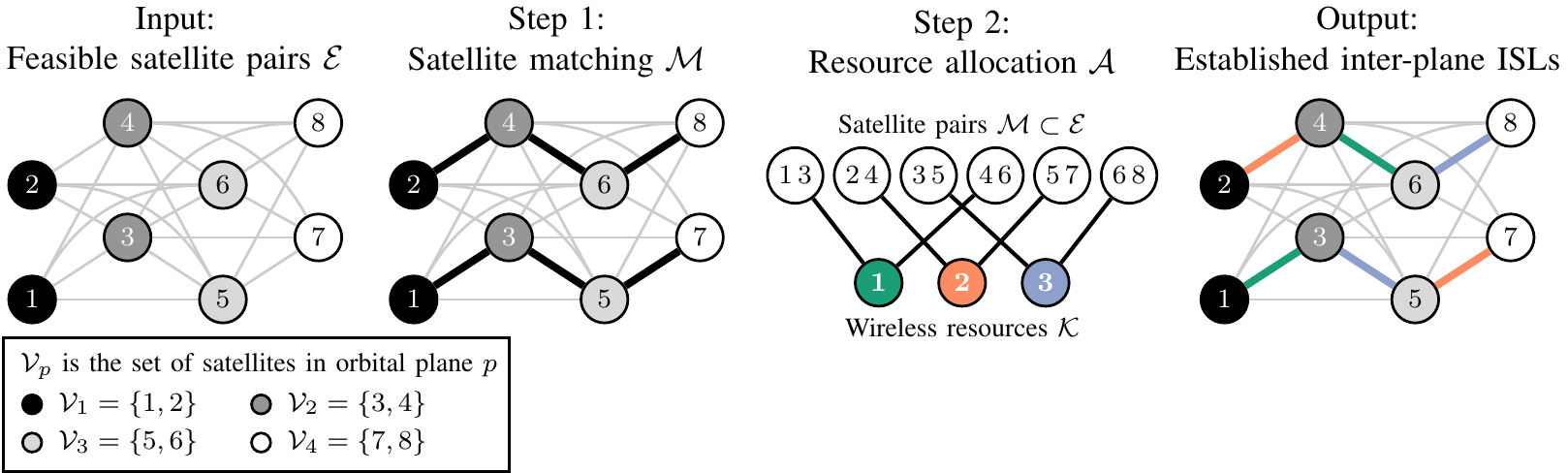}
    \caption{Exemplary diagrams for inter-plane ISL establishment in a satellite constellation where satellites possess two inter-plane transceivers. First, the feasible satellite pairs $\set{E}$ are identified. In this example, \ilm{communication is not feasible between satellites in orbital plane $1$ (black) and those in orbital plane $4$ (white)}. Second, the satellite pairs are selected. \ilm{Finally, interference can be mitigated by allocating orthogonal wireless resources}.} 
    \label{fig:examples}
\end{figure}
 
 \ilm{Note that this and our \ilm{preliminary work~\cite{Soret2019}} are one of the few that focus on the establishment of inter-plane ISLs. Instead, most of the literature focuses on constellation design~\cite{Kak2019}, routing~\cite{Ekici2001}, or on calculating the throughput~\cite{Portillo2019} and \emph{transport capacity} of the constellation \cite{Liu2017, Jiang2020} once the ISLs have been established.} For instance, Jiang and Zhu~\cite{Jiang2020} considered a set of pre-established ISLs and their rates to calculate and allocate the \emph{transport capacity} between source and destination pairs in a scenario with satellites in LEO and higher orbits. Similarly, Liu et al.~\cite{Liu2017} calculated bounds for the transport capacity from the sum of rates in specific cuts in the constellation graph.  
\ilm{Further, Hu et al.~\cite{Hu2020} formulated the problem of rate maximization in an integrated terrestrial and non-terrestrial setup as a competitive market. The latter includes base stations, drones, and one satellite--the service providers--who submit prices for the resources to the user equipments--the consumers~\cite{Hu2020}.}
 
\ilm{Rate maximization in the inter-plane ISLs is an essential step towards maximizing the \emph{transport capacity} in a LEO constellation. Our framework is able to solve the rate maximization problem in any satellite constellation regardless of its geometry and/or symmetry. Therefore, it is also applicable during the initial deployment phases of satellite constellations, where only a few satellites and orbital planes have been deployed and great asymmetries are present. Further, it does not require having different roles in the nodes as in the problem formulated by Hu et al.~\cite{Hu2020}. Finally, our framework can be used in combination with the algorithms proposed by Liu et al.~\cite{Liu2017} and by Jiang and Zhu.~\cite{Jiang2020} to calculate and distribute the transport capacity. In particular, the main contributions of this paper include:}
\begin{enumerate}
\item A detailed analysis of the inter-plane connectivity in satellite constellations. We consider the hardware limitations of small satellites and the impact on the system performance of using different multiple access methods, namely orthogonal frequency division multiple access (OFDMA) and code division multiple access (CDMA). The inter-plane ISLs are characterized in terms of communication range, achievable data rates, propagation delays, and the potential impact of interference. 
\item A framework to find a near-optimal solution to the ISL establishment problem. That is, finding the optimal solution to the ISL establishment is intractable. Instead, we decompose the problem into two tractable sub-tasks: 1) maximum weighted (i.e., data rate) satellite matching and 2) resource allocation. 
    \item Efficient implementations for two different greedy algorithms for maximum weighted satellite matching. 
    \item A specific implementation for a resource allocation algorithm that achieves near-optimal performance by making the best global decision at each iteration. \ilm{Such algorithm is essential for rate maximization whenever the ISLs are affected by interference due to the use of wide-beam antennas.} 
\end{enumerate}

 
  Our results show that: 1) the greedy satellite matching algorithms increases the sum of rates in the constellation by up to $115$\% and by less than $46$\% when compared to our benchmark algorithm, adapted from the work by Ekici \emph{et al.}~\cite{Ekici2001}; 2) the proposed resource allocation algorithm increases the achievable sum of rates by $42$\% with FDMA and by $71$\% with CDMA when compared to random allocation; 3) the maximum achieved sum of rates with FDMA is $84$\% greater than that with CDMA; and 4) given that our matching algorithms are combined with an appropriate constellation design, the propagation delay at $80$\% of the inter-plane ISLs is less than $10$~ms. This goes in line with the technical recommendations from the 3GPP~\cite{3GPPTR38.821}, which considers one-hop propagation delays of $10$~ms as typical in LEO constellations and helps meet the requirements for a one-way link, defined to be $30$~ms~\cite{TR22.822}.
 

The rest of the paper is organized as follows. Section~\ref{sec:systemmodel} presents the system model for the considered LEO constellation and the problem formulation. Then, Section~\ref{sec:matching} presents our proposed framework and a detailed description and analysis of the proposed matching and resource allocation algorithms. Section~\ref{sec:isl_conn} presents an analysis on inter-plane connectivity that serves as a base to select the simulation parameters. Section~\ref{sec:results} presents the results and Section~\ref{sec:conclusions} concludes the paper.



\section{System model and problem formulation} 
\label{sec:systemmodel}
This section presents the system model and formulates the optimization problem. The relevant notation introduced in this section is listed in Table~\ref{tab:notation}.

\subsection{System model}
We consider a general Walker star LEO constellation, such as the one illustrated in Fig.~\ref{fig:constellation}, where $N$ satellites are evenly distributed in $P$ circular and evenly-spaced polar orbital planes. 
The satellites communicate through unicast inter-plane ISLs according to an arbitrary predefined multiple access method. Specifically, in this paper we consider OFDMA and CDMA, but other methods can be implemented. 

\begin{table}[t]
    \centering
     \caption{Notation summary.}
    \begin{tabular}{@{}ll@{}cll@{}}
    \toprule
       Symbol & Description & & Symbol & Description\\\cmidrule{1-2}\cmidrule{4-5}
       $N$ & Total number of satellites & & $P$ & Number of orbital planes\\
        $B$ & System bandwidth &&
        $h_p$ & Altitude of orbital plane $p$\\
        $\epsilon_p$ & Longitude of orbital plane $p$ & & $\theta_u$ & Polar angle of satellite $u$ \\
        $N_p$ & Number of satellites in orbital plane $p$
         & & $p(u)$ & Orbital plane of satellite $u$\\ $\mathrm{R_E}$ & Radius of the Earth  & &
          $Q$ & Number of inter-plane transceivers per satellite\\
        $\set{G}$ & Multi-partite graph of the constellation&&
        $\set{V}$ & Vertex set of $\set{G}$ \\
        $\set{E}$ & Edge set of $\set{G}$ &&
        $\set{K}$ & Set of orthogonal resources \\
        $d\in\{-,+\}$ & Direction in the pitch axis && $d(u,v)$ & Relative direction of satellite $v$ w.r.t. $u$ \\
         $l^\star(uv)$ & Maximum slant range between $u$ and $v$ & &$L(uv)$ & Path loss for an ISL between $u$ and $v$ \\ 
          $K$ & Number of orthogonal wireless resources  &&
          $k$ & Index of an orthogonal wireless resource \\
          $a_{(u,v,k)}$ & Indicator of $u$ transmitting to $v$ with $k$ && $P_\text{out}$ & Outage probability of an ISL\\
         $I(u,v,k)$ & Interference at $v$ for a transmitter $u$ in $k$ && $\set{I}_{(u,v,k)}$ & Interference pattern at $v$ for a transmitter $u$ in $k$ \\
         EIRPG & EIRP plus receiver antenna gain & & $w({e})$ & Weight of edge $e$\\
         $R^\star_\text{SNR}(u,v)\!$ & Maximum rate to transmit from $u$ to $v$ & & $R^\star_\text{SINR}(u,v,k)\!$ & Maximum rate to transmit from $u$ to $v$\\[-0.5em]
         &given $P_\text{out}=0$ and $I(u,v,k)=0$&& &given $P_\text{out}=0$ and $\max_{\set{I}_{(u,v,k)}} I(u,v,k)$\\
         $R_\text{min}$ & Minimum rate to establish an ISL & & $\text{MPL}$ & Maximum FSPL to transmit at $R_\text{min}$\\
         \bottomrule
    \end{tabular}
    \label{tab:notation}
\end{table}
Each orbital plane $p\in\{1,2,\dotsc,P\}$ is deployed at a given altitude above the Earth's surface $h_p$~km, at a given longitude $\epsilon_p$~radians, and consists of $N_p$ evenly-spaced satellites. The polar angle of a satellite $u$ (i.e., its angle w.r.t. the Earth's north pole, the $z$-axis) is denoted as $\theta_u$. Besides, we define the function $p(u)$ to be the orbital plane in which satellite $u$ is deployed. Therefore, the position of satellite $u$ in spherical coordinates is denoted as $(h_{p(u)}+\mathrm{R}_\mathrm{E}, \epsilon_{p(u)}, \theta_u)$, where $\mathrm{R}_\mathrm{E}$ is the Earth's radius.

Each satellite is equipped with a total of four antennas for unicast inter-satellite communication. The intra-plane antennas are located at both sides of the roll axis, whereas the inter-plane antennas are positioned at each side of the pitch axis, as shown in Fig.~\ref{fig:constellation}. We denote the direction of the inter-plane antennas as $d\in\{-,+\}$, where $d=-$ and $d=+$ correspond to the antennas placed at the negative (left) and positive (right) sides of the pitch axes. We consider two cases, where either one or two transceivers, namely $Q\in\{1,2\}$, are available for the inter-plane communication at each satellite. If $Q=2$, every satellite in the constellation has one transceiver per inter-plane antenna, hence, they can establish up to one intra-plane ISL at each side of the pitch axis. On the other hand, if $Q=1$, the satellites can only establish up to one inter-plane ISL in one direction.

We model the constellation at any given time instant $t$ as a weighted undirected graph \mbox{$\set{G}=(\set{V},\set{E})$} where $\set{V}$ is the set of vertices (satellites) and $\set{E}$ is the set of undirected edges (feasible inter-plane ISLs). Throughout the paper, we denote an undirected edge as $uv$ and a source-destination pair as $(u,v)$, where $u,v\in\set{V}$. Graph $\set{G}$ is multi-partite with $P$ vertex classes $\set{V}_1,\set{V}_2,\dotsc,\set{V}_P$. 

Even though $\set{G}$ is a dynamic graph due to the movement of the satellites, we observe the system at specific time instants with a sufficiently short period $T$, and omit the time index $t$ throughout the paper for notation simplicity. A similar approach was employed by Jiang and Zhu~\cite{Jiang2020} and Liu et al.~\cite{Liu2017}. Hence, we denote the weight of an edge $uv\in \set{E}$ as $w(uv)\in\mathbb{R}$. The weight $w(uv)$ is defined as the achievable data rate for the inter-plane ISL between satellites $u$ and $v$ s.t. $u,v\in\set{V}$. In the following, we present the model for inter-plane inter-satellite communication that will serve to calculate these weights. 

Inter-satellite communication occurs in a free-space environment. Therefore, it is mainly affected by the free-space path loss (FSPL) and the (thermal) noise power, which is assumed to be additive white Gaussian (AWGN)~\cite{3GPPTR38.821}. To characterize the inter-plane ISLs, we define the function $d\left(u,v\right)\in\{-,0,+\}$ as the relative direction of satellite $v$ w.r.t. satellite $u$. For the particular case of Walker star constellations, the latter can be obtained by rotating the axes by $-\epsilon_{p(u)}$ along the Earth's rotation axis (i.e., $z$), so that the orbital plane $p(u)$ is positioned along in the $xz$-plane. By doing so, the relative direction can be calculated as the function
\begin{IEEEeqnarray}{rCl}
    f_d(u,v)&=&\sin{\theta_u}\,\sin(\epsilon_{p(u)}-\epsilon_{p(v)}),
\end{IEEEeqnarray}
and we denote the relative direction as
\begin{equation}
    d(u,v) = \begin{cases}
    -, & \text{if } f_d(v,u)> 0\\
    +, & \text{if } f_d(v,u)< 0\\
    0, &\text{otherwise}
    \end{cases}
\end{equation}

In the following, we focus on defining the set of feasible ISLs $\set{E}$ in the graph. That is, the set of satellite pairs where communication is feasible. Naturally, for the intra-plane ISLs we have $d(u,v)=0$ and $p(u)=p(v)$, hence, $uv\notin \set{E}:p(u)=p(v)$. The rest of the feasible satellite pairs are determined by the existence of line-of-sight (LoS) and the magnitude of the Doppler shift.

Let $\lVert uv\rVert$ be the Euclidean distance between two satellites $u$ and $v$. The latter is given in spherical coordinates as
\begin{align}
\lVert uv\rVert =& \Big( (h_p+\mathrm{R_E})^2 + (h_q+\mathrm{R_E})^2 \nonumber\\
&  -2(h_p+\mathrm{R_E}) (h_q+\mathrm{R_E})\left(\cos\theta_{u}\cos\theta_{v} 
+\cos(\epsilon_p-\epsilon_q)\sin\theta_{u}\sin\theta_{v}\right) \Big) ^{1/2},
 \label{eq:distance}
\end{align}
where $p=p(u)$, $q=p(v)$, and $\mathrm{R_E}$ is the radius of the Earth.

Next, we calculate $l^\star(p,q)$, defined as the maximum slant range (i.e., line-of-sight distance) between two satellites $u$ and $v$, in orbital planes $p=p(u)$ and $q=p(v)$, due to the presence of the Earth. That is, the Earth blocks the LoS between $u$ and $v$ if $\lVert uv \rVert>l^\star(p,q)$. Assuming the Earth is perfectly spherical, the latter is given as
\begin{IEEEeqnarray}{rCl}
    l^\star(p,q)&=& \sqrt{h_p(h_p+2\mathrm{R_E})}
    +\sqrt{h_q(h_q+2\mathrm{R_E})}.
    \label{eq:max_l_LoS}
\end{IEEEeqnarray}
Hence, the set of edges with no line of sight (NLoS) is 
$\left\{uv:\lVert uv \rVert>l^\star(p,q)\right\}$. Building on this, the FSPL between $u$ and $v$ is calculated as
\begin{equation}
L(uv)=\begin{cases}
\left(\frac{4\pi \lVert uv\rVert f}{\mathrm{c}}\right)^2, &  \text{if } \lVert uv \rVert \leq l^\star(p,q)\\
\infty & \text{otherwise}.
\end{cases}
\label{eq:pathloss}
\end{equation}
where $f$ is the carrier frequency and $\mathrm{c}=2.998\cdot10^8$~m/s is the speed of light.

Throughout this paper we assume that the inter-plane transceivers are designed to compensate for the Doppler shift in every inter-plane ISL except in the cross-seam ISLs. This is a common practice that has been adopted in commercial deployments such as in the upcoming Kepler constellation~\cite{KeplerIEEESpectrum}, which has a Walker star geometry, and in the literature~\cite{Liu2017}.

In Walker star constellation such as the one illustrated in Fig.~\ref{fig:constellation}, the cross-seam ISLs occur between satellites in orbital planes $1$ (i.e, first) and $P$ (i.e., last), where the relative velocities between the two satellites are close to twice their orbital speeds. As a reference, we calculated a maximum Doppler shift of $114.32$~kHz in these cross-seam ISLs with $f=2.4$~GHz and $P=5$. In comparison, the maximum Doppler shift for ISLs between $p(u)=1$ and $p(v)=2$ is only $36.99$~kHz. Hence, the set of non-feasible edges due to a high Doppler shift is $\left\{uv: \left|p(u)-p(v)\right|= P-1\right\}$.

 Building on the FSPL and the Doppler shift constraints, the set of feasible edges is
\begin{equation}
    \set{E}=\left\{uv\in\set{V}^{(2)}:|p(u)-p(v)|\notin \{0,P-1\}, L(uv)<\infty\right\}.
    \label{eq:edges}
\end{equation}

Having defined the set of feasible edges, we move on to characterize the ISLs. For this, let $G^{d(u,v)}_u$ denote the normalized gain of the inter-plane antenna of satellite $u$ in the direction of $v$ w.r.t. the pitch axis. The antenna gain is a function of the beamwidth and the relative direction of the main lobe of antenna $d(u,v)$ in $u$ w.r.t. the position of $v$. Throughout this paper, we assume that the transmission power $P_t$ is fixed for all satellites and that these are equipped with similar (directional) antennas with perfect beam steering capabilities. Therefore, we simply define the peak gain of both antennas of satellite $u$ in the direction of their main lobe as 
\begin{equation}
    G_\text{max}=\max_{f_d(u,v)} G^{d(u,v)}_u.
\end{equation}
 Furthermore, a satellite pair $uv\in \set{E}$ can only communicate with the antennas in the corresponding directions. Therefore, the transmitter and receiver antennas are always aligned in the direction of maximum radiation so the gain of each of these is $G_\text{max}$. For simplicity, we define the 
 effective isotropic radiated power plus receiver antenna gain (EIRPG) as
\begin{equation}
    \text{EIRPG}=P_tG_\text{max}^2.
\end{equation}  
Naturally, for isotropic antennas we have $G^d_{(u,v)}=1$ for all $d$, and $u,v\in\set{V}$. 

By assuming the wireless channels are symmetric, we define the signal-to-noise ratio (SNR) for an ongoing transmission from $u$ to $v$ and vice-versa as
\begin{equation}
    \text{SNR}(uv)=
    \frac{\text{EIRPG}}{k_\mathrm{B} \tau B  L(uv)},
    \label{eq:snr}
\end{equation}
where $k_\mathrm{B}$ is the Boltzmann constant, $\tau$ is the thermal noise in Kelvin, and $B$ is the channel bandwidth in Hertz.

Next, let $R(u,v,k)\in \mathbb{R}$ be the data rate used for communication from $u$ to $v$ with resource $k$. The latter is selected from an infinite set of possible rates to have zero outage probability $P_\text{out}=0$. Note that, ensuring $P_\text{out}=0$ is of utmost importance in satellite communications to avoid the use of feedback with high round-trip-times (RTTs) due to the long propagation delays. Note that interference can be avoided if the inter-plane ISL antennas in all the satellites combine sufficiently narrow beams with precise beam steering or antenna pointing capabilities. Therefore, the maximum data rate that $u$ can select to communicate with $v$ to ensure $P_\text{out}=0$ in an interference-free environment is
\begin{equation}
    R^\star_\text{SNR}(uv)=\max R\Big(u,v,k\mid I\left(u,v,k\right)=0, P_\text{out}=0\Big)=B \log_2\Big(1+\text{SNR}(uv)\Big).
    \label{eq:snr_rate}
\end{equation}

However, in a general scenario, interference between ISLs is expected to occur due to the sharing of wireless resources by multiple inter-plane ISLs. In the following, we calculate the impact of interference and resource sharing in the achievable data rates at the ISLs.
Recall that, in our setting, a wireless resource $k$ is selected from a pool of $K$ orthogonal wireless resources $\set{K}=\{1,2,\dotsc,K\}$ to be used for communication at each inter-plane ISL (i.e., for each satellite pair). These resources are considered to be either orthogonal sub-carriers in OFDMA or orthogonal codes in CDMA. In the case of OFDMA, we assume that the $K$ sub-carriers are sufficiently close in the frequency domain and we can remove the frequency dependency in the path loss. Hence, we simply denote the path loss between $u$ and $v$ with any $k\in \set{K}$ as $L(uv)$. Besides, these resources are orthogonal to those used for the intra-plane ISLs. Therefore, no interference between intra- and inter-plane ISLs can occur.

Let $\left\{a_{(u,v,k)}\right\}$ be a set of indicator variables s.t. $a_{(u,v,k)}=1$ if satellite $u$ has an ongoing inter-plane transmission to $v$ with resource $k$ and $a_{(u,v,k)}=0$ otherwise. Only one resource is allocated per satellite pair (i.e., ISL), hence, $\Pr\left[a_{(u,v,k)}+a_{(v,u,k)}>1\right]=0$. That is, each of the established ISLs can only be used for transmission by up to one satellite at the same time.

 A physical interference model with constant noise power, based on the power capture model, is considered~\cite{Cardieri2010}. We denote $\set{I}_{(u,v,k)}=\{a_{(i,j,k)},a_{(j,i,k)}:ij\in \set{V}^{(2)}\setminus uv\}$ to be a permissible interference pattern for a transmission from $u$ to $v$ with resource $k$. The latter, whose formal definition is given at the end of this section after introducing the required notation, serves to define the interference at $v$ for an ongoing transmission from $u$ with resource $k$. That is, the interference at $v$ given $a_{(u,v,k)}=1$, as
\begin{IEEEeqnarray}{rCl}
    I(u,v,k) &=& 
    \!\sum_{a_{(i,j,k)}\in \set{I}_{(u,v,k)}} \!\frac{a_{(i,j,k)} P_t\,G^{d(i,j)}_i G^{d(v,u)}_v}{L(iv)}.\IEEEeqnarraynumspace
    \label{eq:interference}
\end{IEEEeqnarray}
Then, the signal-to-interference-plus-noise ratio (SINR) for an ongoing transmission from $u$ to $v$ with resource $k$ is
\begin{IEEEeqnarray}{rCl}
   \text{SINR}(u,v,k)&=& 
   \frac{\text{EIRPG}}{L(uv)\left(k_\mathrm{B} \tau B + I(u,v,k)\right)}.\IEEEyesnumber
    \label{eq:sinr}
\end{IEEEeqnarray}

Throughout this paper, we assume that the interference can be treated as AWGN~\cite{Cardieri2010}. Naturally, if the instantaneous values of all the elements in $\set{I}_{(u,v,k)}$ are known, the maximum data rate at which $u$ can transmit to $v$ can be selected to ensure $P_\text{out}=0$ as
\begin{equation}
    \max R\Big(u,v,k\mid \set{I}_{(u,v,k)}, P_\text{out}=0\Big)=B \log_2\Big(1+\text{SINR}\left(u,v,k\mid \set{I}_{(u,v,k)}\right)\Big).
    \label{eq:tdr}
\end{equation}
However, selecting and achieving the data rate described by~\eqref{eq:tdr} in practice is infeasible as it requires 1) instantaneous and perfect knowledge of the interference, determined by the activity of all the ISLs sharing a specific resource $k$, and 2) real-time and perfect adaptation of the rate. Instead, we consider a realistic scenario in which the rates are selected at the time the ISLs are established to achieve $P_\text{out}=0$ for the maximum interference that can be created by a permissible interference pattern $\set{I}_{(u,v,k)}$.  
Formally, the rates are selected as
\begin{IEEEeqnarray}{rCl}
    R^\star_\text{SINR}(u,v,k)&=&\max R\left(u,v,k\,\middle|\, \max_{\set{I}_{(u,v,k)}}I(u,v,k), P_\text{out}=0\right)\IEEEnonumber\\
    &=&B \log_2\left(1+\min_{\set{I}_{(u,v,k)}} \text{SINR}\left(u,v,k\right)\right).
    \label{eq:sinr_rate}
\end{IEEEeqnarray}

At this point, it is convenient to introduce $R_\text{min}$, defined as the minimum acceptable rate to establish an ISL prior to resource allocation. That is, an ISL between $u$ and $v$ can only be established if $R^\star_\text{SNR}(uv)>R_\text{min}$. The latter represents, for example, the minimum rate required to complete the necessary handshakes between the satellite pairs. Building on this, we calculate the minimum SNR to establish an ISL for all satellites as $\gamma=2^{R_\text{min}/B}-1$. Hence, an ISL between $u$ and $v$ can only be established if 
\begin{equation}
   \text{EIRPG}\geq L(uv) k_\text{B}\tau B \left(2^{R_\text{min}/B}-1\right)
    \label{eq:eirp}
\end{equation}
This allows us to treat the maximum path loss (MPL) to achieve the desired $R^\star_\text{SNR}(uv)=R_\text{min}$ as a design parameter. The latter is given as
\begin{equation}
   \text{MPL}(uv)= \frac{\text{EIRPG}}{k_\text{B}\tau B \left(2^{R_\text{min}/B}-1\right)}.
\end{equation}

\subsection{Problem formulation}
We consider the problem of maximizing the sum of rates in the inter-plane ISLs of a satellite constellation. In the following, we provide essential definitions that will be used to define our problem in a graph setting.

The neighborhood of vertex $u\in \set{V}$ in graph $\set{G}$, denoted as $\Gamma_\set{G}(u)$ is the set of vertices in $\set{V}\setminus \{u\}$ that are adjacent to $u$. The degree of vertex $u\in\set{V}$ in graph $\set{G}$ is $\deg_\set{G}(u)=|\Gamma_\set{G}(u)|$. Furthermore, we denote the number of vertices in the neighborhood of $u$ that are in direction $d$ from $u$ as $\deg^d_\set{G}(u)=|\left\{v\in\Gamma_\set{G}(u):d(v,u)=d\right\}|$ for $d\in\{-,+\}$. The maximal and minimal degrees of a vertex in graph $\set{G}$ is denoted as $\Delta(\set{G})=\max_u \deg_\set{G}(u)$ and $\delta(\set{G})=\min_u \deg_\set{G}(u)$, respectively. The size of a graph $\set{G}=(\set{V},\set{E})$ is given by the number of elements in its edge set $|\set{E}|$. The weight of a set of edges $\set{E}$ is defined as $w(\set{E})=\sum_{e\in{\set{E}}}w(e)$. A matching $\set{M}$ is a subset of $\set{E}$ that represents an association between the vertices in $\set{V}$. A maximum weighted matching is the subset $\set{M}^\star\subseteq \set{E}$ s.t. $w(\set{M})$ is maximal. That is, $\nexists \set{M}':w(\set{M}')>w(\set{M}^\star)$.

 Building on the latter, our problem is formulated as a many-to-one maximum weighted matching problem with \emph{externalities} in a bipartite graph  $\set{G}_\set{A}=(\set{M}\cup\set{K},\set{A})$, where each $uv\in\set{V}_\set{M}\subseteq\set{E}$ is a feasible satellite pair in the constellation $\set{G}=(\set{V},\set{E})$, where $u,v\in\set{V}$. Therefore, each edge $uvk\in\set{A}$ indicates that resource $k$ has been allocated to the communication between satellites $u$ and $v$. The same resource $k$ can be allocated to (i.e., shared by) several satellite pairs $uv\in\set{V}_\set{M}$ if $|\set{V}_\set{M}|>K$. 
 Note that $\set{A}$ is not necessarily a perfect matching, since some satellite pairs may not be allocated resources for communication. 

As defined by~\eqref{eq:snr_rate} and~\eqref{eq:sinr_rate}, we consider a realistic worst case scenario where neither the centralized entity nor the satellites have instantaneous knowledge of $\mathcal{I}_{(u,v,k)}$. Therefore, the rates at each ISL are selected to ensure $P_\text{out}=0$ for any permissible interference pattern and at all times. The latter is defined in the following.
\begin{definition}[Permissible interference pattern]
We define $\mathcal{I}_{(u,v,k)}$ to be a permissible interference pattern for an ongoing transmission from $u$ to $v$ with resource $k$ (i.e., given $a_{(u,v,k)}=1$) if and only if up to one satellite for each satellite pair that has been allocated the same resource $k$, $ijk\in \set{A}\setminus \{uvk\}$ is transmitting at the same time as $u$. Formally,
\[
\mathcal{I}_{(u,v,k)}=\left\{a_{(i,j,k)},a_{(j,i,k)}:ijk\in \set{A}\setminus\{uvk\} \textsc{ and }  a_{(i,j,k)}+a_{(j,i,k)}\in\{0,1\}\right\}.
\]
\end{definition}
Building on this, we define the edge weights of $\set{G}_\set{A}$ as the sum of the rates selected for communication at each satellite, namely \begin{equation}
w(uvk)=R^\star_\text{SINR}(u,v,k)+R^\star_\text{SINR}(v,u,k).
\end{equation}
Note that the weights of each of the edges in $\set{A}$ are affected by the other edges in $\set{A}$. In other words, having a resource $k$ allocated to a satellite pair $uv$ implies that the communication between $u$ and $v$ creates interference to every $ijk\in\set{A}\setminus \{uvk\}$. Therefore, the externalities are represented by the interference created by the sharing of resources, which changes every time a new edge is added to $\set{A}$. Hence, the rates for communication $\{R^\star_\text{SINR}(u,v,k)\}$ can only be selected after $\set{A}$ has been populated.

The optimal solution to our resource allocation problem is the maximum matching $\set{A^\star}$. The latter can only be achieved by adapting the resources allocated to the satellite pairs to the changes in the permissible interference pattern throughout the matching and, in turn, in the rates selected at each edge in $\set{A^\star}$. Clearly, an algorithm capable of finding $\set{A^\star}$ would be of a tremendous complexity. Tn the following, we present our framework to find a near-optimal solution to the ISL establishment problem with a relatively low asymptotic complexity.

\section{Proposed framework}
\label{sec:matching}
\ilm{Our framework solves the ISL establishment problem by first solving the inter-plane satellite matching from a set of feasible edges $\set{E}$. Then, if the ISLs are affected by interference, orthogonal resource allocation takes place. That is, selecting all the satellite pairs first and, if needed, allocate the wireless resources afterwards}.

\subsection{Satellite matching}
 The satellite matching with $Q$ transceivers is a many-to-many maximum weighted matching problem. The goal is to find a subset of edges $\set{M}\subseteq \set{E}$ to maximize the sum of SNR rates in the constellation. For this, we define the weighted subgraph $\set{G}_\set{M}=\left(\set{V}_\set{M}\subseteq\set{V}, \set{M}\right)\subseteq \set{G}$. The latter has to fulfill the following conditions.
 \begin{itemize}
     \item The maximal degree of any vertex in $\set{G}_\set{M}$ is less than or equal to the number of transceivers. That is, $\Delta(\set{G}_\set{M})\leq Q$. Note that we have a one-to-one matching for $Q=1$. 
     \item A satellite $v$ is in the neighborhood of $u$ if and only if there is no other satellite adjacent to $u$ in direction $d(u,v)$. That is, $v\in \Gamma_{\set{G}_\set{M}}(u)$ if and only if $d(u,v)\neq d(u,i)$ for all $i\neq v\in \Gamma_{\set{G}_\set{M}}(u)$
 \end{itemize}
 
 Therefore, it is now convenient to introduce tho following definition.
 \begin{definition}[Permissible neighborhood] We define $\Gamma_{\set{G}_\set{M}}(u)$ to be a permissible neighborhood for vertex $u$ in graph $\set{G}_\set{M}$ if and only if $\deg_{\set{G}_\set{M}}(u)\leq Q$ and $\deg^d_{\set{G}_\set{M}}(u)\in\{0, 1\}$ for all $d\in\{-,+\}$.  \label{def:neighs}\end{definition}
 Naturally, the satellite pairs in $\set{G}_\set{A}$ defined in the previous section must also be selected from a permissible neighborhood.

 As described in Section~\ref{sec:systemmodel}, an interference-free environment is considered during the satellite matching and, since the wireless channels are symmetric, the weights of the edges are defined as $w(uv)=2R^\star_\text{SNR}(uv)$. Consequently, the satellite matching problem is defined as
\begin{IEEEeqnarray*}{lCll}
    \text{maximize}&\quad& w(\set{M})=2\sum_{uv\in \set{M}}R^\star_{\text{SNR}}(uv)\\
    \text{subject to}&&  \deg^d_{\set{G}_\set{M}}(u)\in\{0,1\},& \forall\, u\in \set{V}, d\in\{-,+\}\\
    && \Delta(\set{G}_\set{M}) \leq Q.\IEEEyesnumber
    \label{eq:isl_matching_prob}
\end{IEEEeqnarray*}

If no interference can occur between the selected satellite pairs, solving~\eqref{eq:isl_matching_prob} solves the ISL establishment problem, as any number of resources $K\in\mathbb{N}^+$ allows to directly use the rates $R^\star_\text{SNR}(uv)$ for communication with $P_\text{out}=0$. Otherwise, resource must take place as described in the following to maximize the sum of rates in the constellation.

\subsection{Resource allocation}
Resource allocation is still the many-to-one maximum weighted matching problem with externalities in a bipartite graph, described above. However, its complexity is greatly reduced by having a fixed set of satellite pairs $\set{M}$.

Building on this, we formally define the resource allocation problem as
\begin{IEEEeqnarray*}{lCll}
    \text{maximize}&\quad& \IEEEeqnarraymulticol{2}{l}{w(\set{A})=\sum_{uvk\in\set{A}}
    R^\star_\text{SINR}(u,v,k)+R^\star_\text{SINR}(v,u,k)}\\
     \text{subject to}&&  
     \deg_{\set{G}_\set{A}}(uv)=1, \qquad& \forall\, uv\in\set{M}.
    \IEEEyesnumber
    \label{eq:wa}
\end{IEEEeqnarray*}
Note that the satellite pairs in $\set{M}$  fulfill the conditions stated on Definition~\ref{def:neighs}.

 Next, let $M=|\set{M}|$ and $A=|\set{A}|$ be the size of the satellite matching and resource allocation graphs, respectively. Observe that the constraint in~\eqref{eq:wa} ensures that $A=M$ at the end of the resource allocation. Thus, a perfect matching is guaranteed.

The satellite matching and resource allocation phases of our framework are illustrated in Fig.~\ref{fig:examples} on page~\pageref{fig:examples} for  $Q=\{1,2\}$ and $K=3$ in a region of a Walker star constellation with $P=4$.


\section{Algorithms}
\label{sec:matching_algs}

Our proposed algorithms are centralized. Therefore, an entity with the orbital parameters of all the satellites is in charge of computing a new matching every $T$~seconds. This entity can be deployed either at the ground segment (e.g., a ground station) or at the space segment (e.g., a satellite with sufficient computational power). Note that, for the centralized matching to be feasible, the matching must be solved ahead of time and communicated to the whole constellation. Hence, the processing time and the communication overhead to deliver the result to all the satellites must be taken into account, so the satellites can complete the process in a timely manner. Because of this, the status of the buffer and, hence, the activation pattern of the satellites cannot be taken into account in real time. A solution to incorporate the activation pattern of the satellites is having them--or even the centralized entity--to make projections and infer the activity of the satellites ahead of time. However, these approaches are out of the scope of this paper. Instead, we assume a general case in which the centralized entity has no prior knowledge on the activation pattern of the satellites.

We propose the use of greedy algorithms to solve, first, the satellite matching and, then, the resource allocation. Greedy algorithms make the best decision at each iteration and usually present a relatively low complexity. On the downside, greedy algorithms are not guaranteed to converge to the optimal solution in matching problems. However, the worst-case performance of greedy algorithms in one-to-one weighted matching problems is well characterized~\cite{Avis1983}. 

\begin{theorem}[Worst-case result for greedy algorithms in a one-to-one maximum weighted matching] Let $\set{M}^\star$ be the optimal (i.e., maximum weighted) matching on an undirected graph $\set{G}=(\set{V},\set{E})$ and $\set{M}$ be the matching with a greedy algorithm that selects the maximum weighted edge at each iteration. The worst-case performance of such greedy algorithm w.r.t. the optimal matching is given as
\begin{equation}
        w(\set{M})\geq \frac{w(\set{M}^\star)}{2}
    \end{equation} 
    \end{theorem}
    \begin{proof} Let  $\set{L}=(\ell_1,\ell_2,\dotsc, \ell_{|\set{E}|})$ be the ordered list s.t. $w(\ell_i)\geq w(\ell_{i+1})$ for all $i\in\{1,2,\dotsc,|\set{E}|-1\}$. At each iteration of the greedy algorithm, the first edge in the list, namely $\ell_1$, is added to $\set{M}$ and all the incident edges are deleted from $\set{L}$. Therefore, at most, two edges in the maximum matching, namely $e_1,e_2\in\set{M}^\star$, are removed from $\set{L}$ at each iteration. Since $w(\ell_1)\geq\max\{w(e_1),w(e_2)\}$, we have that $2w(\ell_1)\geq w(e_1)+w(e_2)$. This scenario can occur $|\set{M}^\star|/2$ times, but the greedy algorithm will continue adding edges until $\set{L}=\emptyset$. Since $w(\set{M})=\sum_{e\in\set{M}}w(e)$, we have that $2w(\set{M})\geq w(\set{M}^\star)$. This concludes the proof.
    \end{proof}

In the following, we describe and derive the complexity of the centralized algorithms to solve the satellite matching and resource allocation. These can be used for any constellation geometry, value of $P\in\mathbb{N}^+$, and $Q\in\{1,2\}$ derive their complexity.

\subsection{Satellite matching}
\textbf{Greedy Independent Experiments satellite Matching (GIEM):} This is a greedy centralized matching algorithm, where the matching $\set{M}$ is solved every time from $\set{M}=\emptyset$. 

An efficient implementation is to
create an ordered queue $\set{L}=(\ell_1, \ell_2,\dotsc)$ with elements $\ell_i\in \set{E}$ s.t. $w(\ell_i)=R^\star_\text{SNR}(\ell_i)$, $w(\ell_i)\geq w(\ell_{i+1})$, and $w(\ell_i) \geq R_\text{min}$ for all $i$. Then, at each iteration, the first element $\ell_1=uv$ in the queue $\mathcal{L}$ is first added to $\set{M}$ if $\deg_{\set{G}_\set{M}}^{d(u,v)}(u)+\deg_{\set{G}_\set{M}}^{d(v,u)(v)}=0$, $\deg_{\set{G}_\set{M}}(u)<Q$, and $\deg_{\set{G}_\set{M}}(v)$. Then, $\ell_1$ is removed from $\mathcal{L}$ This process is repeated until the queue $\mathcal{L}$ is empty. Algorithm~\ref{alg:independent_experiments} summarizes the GIEM algorithm. 

\begin{algorithm} [t]
	\centering
	\caption{Algorithm for greedy independent experiments satellite matching (GIEM).}
	\begin{algorithmic}[1] 
	\renewcommand{\algorithmicrequire}{\textbf{Input:}}
		\renewcommand{\algorithmicensure}{\textbf{Output:}}
		\REQUIRE Set of feasible weighted edges $\set{E}$
		\REQUIRE Number of transceivers $Q$
		\STATE $\set{G}_\set{M}=(\set{V}_\set{M}=\emptyset,\set{M}=\emptyset)$
		\STATE Create $\mathcal{L}=(\ell_1,\ell_2,\dotsc)$ with $\{\ell_i\}=\set{E}$ s.t. $w(\ell_i)\geq w(\ell_{i+1})$, and $w(\ell_i) \geq R_\text{min}$  for all $i$.
		\WHILE { $\mathcal{L} \neq \emptyset$}
		\STATE $uv\leftarrow\ell_1$
		\IF {$\deg_{\set{G}_\set{M}}^{d(u,v)}(u)+\deg_{\set{G}_\set{M}}^{d(v,u)}(v)==0$ \AND $\deg_{\set{G}_\set{M}}(u)<Q$ \AND $\deg_{\set{G}_\set{M}}(v)<Q$ }
		\STATE $\set{M}\leftarrow \set{M}\cup \{uv\}$
		\STATE  {$\set{V}_\set{M}\leftarrow\set{V}_\set{M}\cup\{u,v\}$}
		\ENDIF
		\STATE Delete $\ell_1$
		\ENDWHILE
	\ENSURE \ilm{$\set{M}$}
	\end{algorithmic}  \label{alg:independent_experiments}
\end{algorithm}

To calculate the complexity of the GIEM algorithm, we first define \[\set{E}'=\{uv\in \set{V}^{(2)}:p(u)\neq p(v)\}\supset\set{E}.\] Next, we calculate an upper bound for the number of feasible edges in $\set{L}$.  
\begin{equation}
   |\set{L}|=|\set{E}|\ll|\set{E}'|=\frac{1}{2}\left[N^2-\sum_{p=1}^P N_p^2\right].
\end{equation}
Note that, for the case with $N_p=N/P$, we have $|\set{E}'|=PN_p^2(P-1)$. The insertion of the elements in $\mathcal{L}$ has a  cost $\mathcal{O}\left(|\set{L}|\right)\leq \mathcal{O}\left(|\set{E}'|\right)$. Then, the cost of sorting the list $\mathcal{L}$ is  $\mathcal{O}\left(|\set{L}|\log_2\left(|\set{L}|\right)\right)$ with Merge Sort.

Then, at each iteration, two comparisons, one deletion (the first element in $\mathcal{L}$), up to one insertion in $\set{M}$, and up to two insertions in $\set{V}_\set{M}$ are performed. All of these operations have a cost $\mathcal{O}(1)$ and the process is repeated $|\set{L}|$ times. Therefore, the cost of the operations performed after the list $\set{L}$ has been sorted is $\mathcal{O}\left(|\set{L}|\right)\leq\mathcal{O}(|\set{E}'|)=\mathcal{O}\left(PN_p^2(P-1)\right)$, which for $P=2$ is one order of magnitude lower than the complexity of the well-known Hungarian algorithm, namely $\mathcal{O}\left(N_p^3\right)$. However, the overall cost of the GIEM algorithm is determined by the sorting of the list $\mathcal{O}\left(|\set{L}|\log_2\left(|\set{L}|\right)\right)$. 

It is easy to see that the solution provided by the GIEM algorithm is unique if a correct and consistent sorting algorithm is used, creating the exact same list $\set{L}$ from the exact same input $\set{E}$. Finally, the GIEM algorithm is guaranteed to terminate after $|\set{L}|+1$ executions of line $3$. This is because exactly one element from $\set{L}$ is removed at each iteration independently on whether it is added to $\set{M}$ or not and $|\set{L}|\leq|\set{E}|\leq\infty$.

 Furthermore, avoiding \emph{unstable pairs} in the matching $\set{M}$ is a common feature among greedy algorithms. An unstable a pair in a matching $\in\set{M}$ is a vertex pair $u$ and $v$ s.t. $\{uj,iv\}\in\set{M}$ even though the edge $uv$ is preferred by both $u$ and $v$. To illustrate that the GIEM algorithm avoids unstable pairs, consider satellites $u$ and $v$, where $uv\in\set{E}$. The GIEM algorithm will always find the weight with the greatest weight at each iteration, namely $\ell_1$. If $w(uv)>w(ui)$ for all $i\neq v$, $uv$ will be added to $\set{M}$ before any other $ui\in\set{E}$ if doing so maintains permissible neighborhoods for both $u$ and $v$ in $\set{G}_\set{M}$. Thus, no unstable pairs can be created.

Greedy Markovian satellite Matching (GMM): This \ilm{is an extension} of the GIEM matching where the satellite pairs are maintained for as long as possible. Therefore, the satellites always prefer to be paired as in the previous matching and rate maximization is a secondary objective. This considerably reduces the number of changes in the satellite pairs and, hence, the number of handshakes that must be performed at each realization of the matching. Let $\set{M}(n)$ be the $n$th realization of the satellite matching. The starting point for the $n$th realization of the GMM is $\set{M}(n)=\emptyset$ and its previous realization $\set{M}(n-1)$. Then, the algorithm searches within $\set{M}(n-1)$ to identify the satellite pairs that are still feasible. These are the edges in the set $\left\{e\in \set{M}(n-1)\cap \set{E}\right\}$ and includes them, one by one, in $\set{M}(n)$ if the quota of the antennas has not been reached. This can happen if there are changes in the relative direction of the satellites. After \ilm{each pair is added to $\set{M}(n)$, all the pairs that are no longer feasible are removed from $\set{E}$ (see line 6, from Algorithm~\ref{alg:markovian}).  Then, the resulting $\set{E}$ is used as an input to the GIEM algorithm, whose output is $\set{M}_\text{GIEM}$. From there, the output of the GMM algorithm is $\set{M}(n)\cup\set{M}_\text{GIEM}$.}  Algorithm~\ref{alg:markovian} summarizes the GMM algorithm.

\begin{algorithm} [t]
	\centering
	\caption{Algorithm for greedy Markovian satellite matching (GMM).}
	\begin{algorithmic}[1] 
	\renewcommand{\algorithmicrequire}{\textbf{Input:}}
		\renewcommand{\algorithmicensure}{\textbf{Output:}}
		\REQUIRE Set of weighted edges $\set{E}$
		\REQUIRE Number of transceivers $Q$
		\REQUIRE Previous matching $\set{M}(n-1)$
			\STATE $\set{G}_\set{M}=(\set{V}_\set{M}=\emptyset,\set{M}(n)=\emptyset)$
		\FORALL{$uv\in \set{M}(n-1)\cap \set{E}$}
		\IF {$\deg_{\set{G}_\set{M}}^{d(u,v)}(u)+\deg_{\set{G}_\set{M}}^{d(v,u)}(v)==0$ \AND $\deg_{\set{G}_\set{M}}(u)<Q$ \AND $\deg_{\set{G}_\set{M}}(v)<Q$ }
		\STATE $\set{M}(n)\leftarrow \set{M}(n)\cup \{uv\}$
		\STATE {$\set{V}_\set{M}\leftarrow\set{V}_\set{M}\cup\{u,v\}$}
		\STATE \ilm{$\set{E}\leftarrow\set{E}\setminus \big\{\left\{ui\in\set{E}:d(u,i) = d(u,v)\right\}\cup \left\{vi\in\set{E}:d(v,i) = d(v,u)\right\}\big\}$}
		\ENDIF
		\ENDFOR
	\STATE \ilm{Perform the GIEM algorithm (Algorithm~\ref{alg:independent_experiments}) with the resulting $\set{E}$ as an input; get $\set{M}_\text{GIEM}$ as an output} 
		\ENSURE \ilm{$\set{M}(n)\leftarrow\set{M}(n)\cup\set{M}_\text{GIEM}$}
	\end{algorithmic}  \label{alg:markovian}
\end{algorithm}
Note that, for the GMM we have an even smaller $|\set{L}|$ than for the GIEM because the former initiates the search within the previous matching $\set{M}(n-1)$. Therefore, the majority of the reduction in the execution time of the GMM algorithm when compared to the GIEM algorithm, which will be observed in Section~\ref{sec:results}, is due to sorting a smaller list. The properties described above for the GIEM algorithm also hold for the GMM algorithm.

\subsection{Greedy resource allocation (GRA)}
 Once the satellite pairs have been formed, orthogonal wireless resources are assigned to maximize the sum of rates as a function of the SINR (i.e., considering the interference). Hence, the weights at this phase are defined as $w(uvk)=R^\star_\text{SINR}(u,v,k)+R^\star_\text{SINR}(v,u,k)$. Let $\set{A}=\emptyset$ be the resource allocation at the beginning of a realization, that is, immediately after the satellite matching $\set{M}$ (or $\set{M}(n)$ with the GMM algorithm) has been populated.
 At the beginning of the algorithm, the ordered list $\set{L}_\set{M}$ is created with all the elements $uv\in\set{M}$.
Then, at each iteration, the centralized entity selects $uv=\ell_1$ and allocates the resource $k^\star$ that leads to the maximum sum of rates (i.e., the greedy choice made globally). For this, the centralized entity calculates $w(\set{A}\cup \{uvk\})$ for all $k\in\set{K}$ and assigns
\begin{equation}
    k^\star= \argmax_k w(\set{A}\cup \{uvk\})
    \label{eq:kstar}
\end{equation} 
to the satellite pair $uv$, where $w(A)$ is given as in~\eqref{eq:wa}. Recall that $R^\star_\text{SINR}(u,v,k)$ is the maximum rate that can be selected to transmit from $u$ to $v$ with resource $k$ with zero outage probability. The latter is calculated by evaluating the SINR with all possible combinations of $\set{I}_{(u,v,k)}$ via exhaustive search. This process is summarized in Algorithm~\ref{alg:res_alloc}.

\begin{algorithm} [t]
	\centering
	\caption{Centralized algorithm for greedy resource allocation (GRA).}
	\begin{algorithmic}[1] 
	\renewcommand{\algorithmicrequire}{\textbf{Input:}}
		\renewcommand{\algorithmicensure}{\textbf{Output:}}
		\REQUIRE Satellite matching $\set{M}$	
	\REQUIRE Set of resources $\set{K}$
	\STATE $\set{L}_\set{M}=(\ell_1,\ell_2,\dotsc, \ell_M)$ with $\{\ell_i\}=\set{M}$ s.t. $w(\ell_i)\geq w(\ell_{i+1})$ for all $i$.
		\STATE $\set{A}=\emptyset$, 
		\WHILE{$\set{L}_\set{M}\neq\emptyset$}
		\STATE $uv\leftarrow \ell_1$
		\STATE Allocate resource $k^\star$ to $uv$ according to~\eqref{eq:kstar}
		\STATE $\set{A}\leftarrow \set{A}\cup \{uvk^\star\}$
			\STATE Delete $\ell_1$
		\ENDWHILE
		\ENSURE $\set{A}$
\end{algorithmic}
\label{alg:res_alloc}
\end{algorithm}

To obtain the complexity of our GRA algorithm, we have that at most $4K\,(m-1)$ operations are performed at the $m$th iteration of the algorithm to calculate $w(\set{A}\cup \{uvk\}))$ for all the possible values of $k$. This is because there are four possible combinations of values for $\{a_{(i,j,k)}, a_{(j,i,k)}, a_{(u,v,k)}, a_{(v,u,k)}\}$ that can minimize the SINR for each of the links $\{ijk\}\in\set{A}$ at the $m$th iteration and $uvk$. Next, $K$ additions and comparisons are performed to identify $k^\star$. Since a total of $M$ iterations are required, the complexity of the GRA algorithm is $\mathcal{O}\left(KM^2\right)$.

Clearly, the solution with our GRA algorithm is unique if a correct and consistent sorting algorithm is used to create $\set{L}_\set{M}$. Furthermore, the GRA algorithm is guaranteed to provide the optimal solution that can be obtained with any greedy algorithm because exhaustive search is used to select the $k^\star$ that maximizes the sum of rates at each iteration; hence, $k^\star$ is the best global decision. Finally, the GRA algorithm is guaranteed to terminate after $M+1$ executions of line $3$ of Algorithm~\ref{alg:res_alloc}, as exactly one element in $\set{L}_\set{M}$ is deleted at each iteration and $|\set{L}_\set{M}|=M$ before the first execution of the loop.
\section{Experiment design and parameter selection}
 \label{sec:isl_conn}
The algorithms and the analysis presented in the previous section are generic and valid for any graph that represents a communication network topology with analogous constraints as a satellite constellation. In this section, we particularize to representative satellite constellation topologies and investigate the connectivity characteristics of the inter-plane ISLs to select appropriate simulation parameters. 

In particular, we are set to select the minimum value for the EIRPG needed to ensure that $\delta(\set{G})\geq 1$ and, hence, that $\set{G}=(\set{V},\set{E})$ is a connected graph at all times. That is, to ensure that all the satellites have at least one possible inter-plane neighbor with which they can communicate at a rate higher than $R_\text{min}$ at all times. Hereafter, we refer to this characteristic as \emph{full inter-plane ISL connectivity}. Furthermore, we define the performance indicators and briefly describe the satellite matching algorithm used as a benchmark.
 
As a starting point, we define $L^\star(P,N_p,f)$ as the maximum FSPL to the nearest inter-plane neighbor in a Walker star constellation. Next, we set $\text{MPL}=L^\star(P,N_p,f)$ and substitute $L(uv)$ with the latter in~\eqref{eq:eirp} to calculate the minimum EIRPG that ensures full inter-plane ISL connectivity.

 Recall that the polar angle of a satellite $u$ is denoted as $\theta_{u}$. From there, let $v^\star\in \set{V}_q$ be the closest satellite in $q=p(v^\star)$ to $u\in \set{V}_p$. We have that
\begin{IEEEeqnarray}{rCl}
    v^\star &=& \argmin_{v\in\set{V}_q} \lVert uv\rVert \iff \theta_{v^\star}=\left[\theta_u-\pi/N_p,\theta_u+\pi/N_p\right]\qquad \forall p,q\in \{1,2,\dotsc, P\}~\IEEEeqnarraynumspace
\end{IEEEeqnarray}
Therefore, $|\theta_u-\theta_{v^\star}|\in\left[0,\pi/N_p\right]$. Besides, we observe that the difference in longitude between adjacent orbital planes is $\pi/P$. Building on this, we find the maximum slant range between two satellites $u$ and $v^\star$ in adjacent orbital planes, namely $p=p(u)$ and  $q=p(v^\star)=(p+1\bmod P)$ in a general Walker star constellation from~\eqref{eq:distance} as
\begin{IEEEeqnarray}{rCl}
l_\text{adj}^\star(P,N_p) 
&=&\max_{\theta_{u},\,\Delta\theta,\,p} \Big( (h_p+\mathrm{R_E})^2 + (h_q+\mathrm{R_E})^2 -2(h_p+\mathrm{R_E}) (h_q+\mathrm{R_E})\IEEEnonumber\\
& & \times\left(\cos\theta_{u}\cos\left(\theta_{u}+|\theta_u-\theta_{v^\star}|\right)
+\cos\left(\epsilon_p-\epsilon_q\right)\sin\theta_{u}\sin\left(\theta_{u}+|\theta_u-\theta_{v^\star}|\right)\right) \Big) ^{1/2}\IEEEyesnumber
 \label{eq:maxdistance}
\end{IEEEeqnarray}
From there, we introduce specific characteristics of our constellation. In particular, we consider that the lowest orbital plane is deployed at a typical altitude of $h_1=600$~km and an orbital separation of $10$~km between orbital planes. Therefore, the altitude of the orbital planes is given as $h_p=h_1+10(p-1)$~km for all $p\in\left\{1,2,\dotsc,P\right\}$. Besides, we assume that the cross-seam inter-plane ISLs are not implemented. Building on this, using simple optimization techniques and due to the symmetry of the slant range (metric) allows us to obtain the closed-form expression of one of the maxima of \eqref{eq:maxdistance}. For example, a maximum is achieved at $\theta_u^\star=\pi/2$, 
 $\Delta \theta^\star=\pi/N_p$, and $p^\star=P-1$, where we have
\begin{IEEEeqnarray}{rCl}
l_\text{adj}^\star(P,N_p) &=&  \Bigg( \left(h_{P-1}+\mathrm{R_E}\right)^2 + (h_P+\mathrm{R_E})^2 \IEEEnonumber\\ 
&&-2(h_{P-1}+\mathrm{R_E}) (h_P+\mathrm{R_E})\left( 
\cos\left(\frac{\pi}{P}\right)\sin\left(\frac{\pi\left(2+N_p\right)}{2N_p}\right)\right) \Bigg) ^{1/2}.
 \label{eq:maxdistance_b}
\end{IEEEeqnarray}
Then, $L^\star\left(P,N_p,f\right)$ is obtained by substituting $\lVert uv \rVert$ with $l_\text{adj}^\star(P,N_p)$ in~\eqref{eq:pathloss}. Thus, we set
\begin{equation}
    \text{EIRPG}=L^\star(P,N_p,f) k_\text{B}\tau B \left(2^{R_\text{min}/B}-1\right).
\end{equation}
Throughout the rest of the paper we investigate the performance of the satellite matching algorithms in two operation regimes: limited and full inter-plane ISL connectivity. To do so, we fix $f=2.4$~GHz, $R_\text{min}=10$~kbps, and $\text{MPL}=L^\star(7,40,f)$~dB, so that full connectivity is only guaranteed for $P\geq 7$ and conduct our analyses for $P\in\{5,6,7,8\}$; these and other relevant parameters are listed in Table~\ref{tab:param}. The communication range to ensure full-inter-plane connectivity with $P=7$ is $l_\text{adj}^\star(7,40)=3527$~km. 
We calculated that an $\text{EIRPG}\geq3.74$~W is needed achieve this communication range with the selected parameters.

\begin{table}[t]
\centering\caption{Parameter settings for performance evaluation.}
\renewcommand{\arraystretch}{1}
\begin{tabular}{@{}lcl@{}}
\toprule
Parameter & Symbol & Setting\\\midrule
Number of orbital planes & $P$ & $\{5,6,7,8\}$\\
Number of satellites per orbital plane & $N_p$ & $40$\\
Altitude of orbital plane $p$ [km] & $h_p$ & $600 + 10(p-1)$\\
Longitude of orbital plane $p$ [rad] & $\epsilon_p$ & $\pi(p-1)/P$\\
Minimum acceptable rate [kbps] & $R_\text{min}$ & $10$\\
EIRP plus receiver antenna gain & EIRPG & $3.74$~W\\
Carrier frequency in the S-band [GHz]& $f$ & $2.4$\\
Carrier bandwidth & $B$ & $20$~MHz\\
Thermal noise~\cite{3GPPTR38.821} & $\tau$ & $354.81$~K\\
Number of inter-plane transceivers & $Q$ & $\{1,2\}$\\
Matching period [s]  & $T$ & $30$\\ 
\bottomrule
\end{tabular}
\label{tab:param}
\end{table}
In our analyses, we consider that the inter-plane transceivers have no self-interference cancellation capabilities. Therefore we set $L(vv)=1$ for all $v$ to enable the calculation of the interference $I(v,v,k)$ as described in~\eqref{eq:interference}. Besides, we consider two scenarios for the impact of the antenna design on interference. The first one is an optimistic scenario where the antennas have sufficiently narrow beams and perfect beam steering capabilities. Therefore, the power towards the intended receiver is always $\text{EIRPG}=3.74$~W and any interference is avoided (i.e., s.t. $I(u,v,k)$ for all $u$, $v$, and $k$). The second is a worst-case scenario, where the interference caused by isotropic antennas is considered. These scenarios correspond to the tight upper and lower bounds on performance for the conservative $R_\text{min}=10$~kbps and the calculated EIRPG. 

To obtain the results presented in the next section, a simulator of the constellation geometry that implements the algorithms described in Section~\ref{sec:matching_algs} was developed in Python 3 specifically for this task. Monte Carlo simulations were run on a PC with Ubuntu 18.04.2 LTS ($64$~bit), an Intel Core i7-7820HQ CPU, $2.9$~GHz, and $16$~GB RAM, with a clock precision of $10^{-7}$~s. In each experiment, the constellation is first rotated according to the period between consecutive matching realizations (i.e., observations) $T=30$~seconds and, then, the satellite matching and resource allocation algorithms are executed. A total of $N_\text{sim}=1000$ matching periods are simulated, which gives a total simulation period of $30000$~seconds. 
As a reference, the latitude of each satellite changes by less than $0.01\pi$ between two consecutive realizations of the matching algorithms. The $n$th realization of the satellite matching and resource allocation algorithms are hereafter denoted as $\set{M}(n)$ and $\set{A}(n)$, respectively. 

The selected performance indicators to assess the performance of the satellite matching algorithms are the empirical mean number of established inter-plane ISLs per satellite 
\begin{equation}
    \hat{\mu}_{M}=\frac{1}{N_\text{sim}\,N} \sum_{n=1}^{N_\text{sim}} |\set{M}(n)|=\frac{1}{2N_\text{sim}\,N} \sum_{n=1}^{N_\text{sim}} \sum_{u\in\set{M}} \deg_{\set{G}_\set{M}}(u)
\end{equation} 
and the sum of rates as a function of the SNR (i.e., neglecting interference)
\begin{equation}
\mu_{R^\star_\text{SNR}(\set{M})} =\frac{1}{N_\text{sim}}\sum_{n=1}^{N_\text{sim}} w(\set{M}(n))
\end{equation}
Then, the selected performance indicator to assess the performance of the resource allocation algorithm is the ratio of normalized mean sum of rates  
\begin{equation}
\hat{\mu}_{R^\star_\text{SINR}(\set{A})} =\frac{1}{\mu_{R^\star_\text{SNR}(\set{M})}\,N_\text{sim}}\sum_{n=1}^{N_\text{sim}}  w(\set{A}(n)).
\end{equation}

We have selected a \emph{Geographical matching (GEO)} algorithm as benchmark for the satellite matching algorithms. The GEO algorithm is inspired by the routing algorithm provided by Ekici \emph{et al.}~\cite{Ekici2001} where the latitude is divided into $N_p$ regions called \emph{logical locations} of width $2\pi/N_p$ as shown in Fig.~\ref{fig:geo_matching}. Then, satellites in neighbouring orbital planes in the same logical location are matched. 
The GEO algorithm can compute the logical locations of all the $N_p P$ satellites and perform the matching in a single pass. 
Hence, the complexity of the GEO algorithm is $\mathcal{O}(P N_p)$.

Finally, the performance of the GRA algorithm is compared to that of round-robin and random allocation. These three algorithms will be applied after satellite matching with the GIEM algorithm. In the round-robin approach, the $K$ orthogonal wireless resources are allocated one by one from the first to the last element of the ordered list $\set{L}_\set{M}$, with complexity $\mathcal{O}(M)$.

\begin{figure}[t]
\centering
\includegraphics{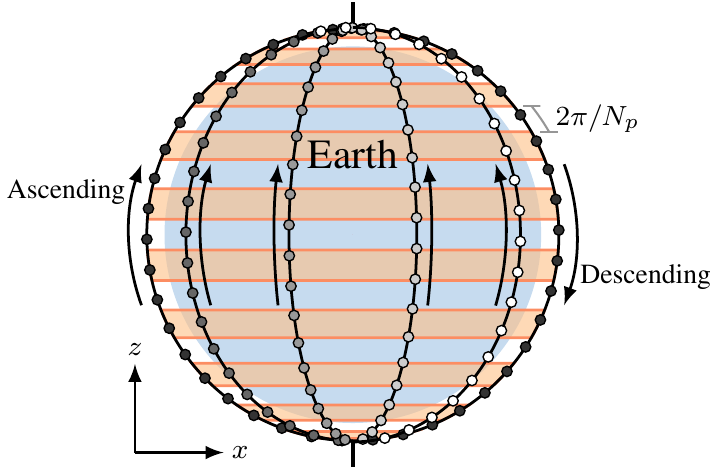}
\caption{Concept behind the GEO matching approach, where logical locations of the satellites are  highlighted~\cite{Ekici2001}.}
\label{fig:geo_matching}
\end{figure}
\section{Results}
\label{sec:results}
This section presents the most relevant results on the performance of the algorithms described in Section~\ref{sec:matching_algs} with the parameters listed in Table~\ref{tab:param}.

As a starting point, we illustrate the characteristics of our problem and the impact of the selected parameters (listed in Table~\ref{tab:param}) in Fig.~\ref{fig:num_matches}, which shows the degree $\deg_{\set{G}_\set{M}}(u)$ for all the satellites in the constellation after a typical realization of the GIEM algorithm for $Q=2$. Note that the degree of satellites $u$ in orbital planes $1$ and $P$ (located in the extremes of each of the subfigures of Fig.~\ref{fig:num_matches}) is lower since cross-seam ISLs are not implemented. It is easy to see that a greedy algorithm will start by establishing the ISL around the crossing points of the orbital planes (i.e., near the  poles), where the shortest slant ranges occur. On the other hand, full inter-plane ISL connectivity is observed for $P\geq7$ because this was the value of $P$ used to calculate the EIRPG as in~\eqref{eq:maxdistance_b}. Nevertheless, full inter-plane ISL connectivity is not necessary to reap some of the benefits of the inter-plane ISLs. For example, it can be seen in Fig.~\ref{fig:num_matches_b} that any satellite is within a few intra- and inter-plane ISL hops from each other. That is, if a direct inter-plane ISL is not available, the packets can be first routed through intra-plane ISLs towards the poles until an inter-plane ISL is available.
\begin{figure}[t]
    \centering
    \subfloat[]{\includegraphics{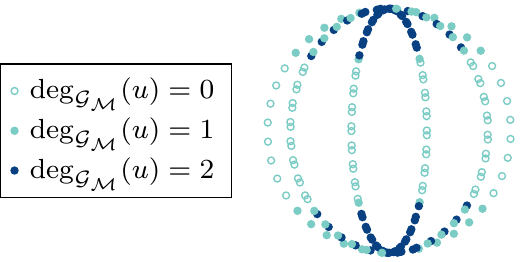}\label{fig:num_matches_a}}\hfil
    \subfloat[]{\includegraphics{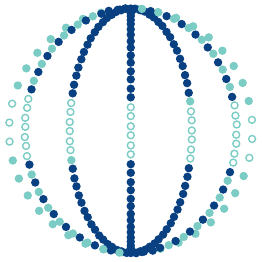}\label{fig:num_matches_b}}\hfil
    \subfloat[]{\includegraphics{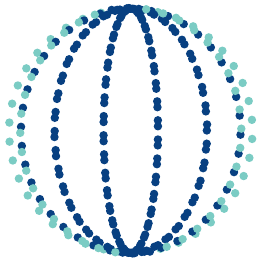}\label{fig:num_matches_c}}\hfil
    \subfloat[]{\includegraphics{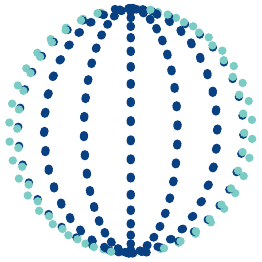}\label{fig:num_matches_d}}
    \caption{Frontal view of the constellation showing the number of matches per satellite after one representative realization of the GIEM algorithm given $\text{MPL}=L^\star(7,40,2.4)$ and $Q=2$ with (a) $P=5$, (b) $P=6$, (c) $P=7$, and (d) $P=8$.}
    \label{fig:num_matches}
\end{figure}

Next, we illustrate the performance of the satellite matching algorithms in terms of $\hat{\mu}_{M}$ in Fig.~\ref{fig:mean_ISLs} and of $\mu_{R^\star_\text{SNR}(\set{M})}$ in Fig.~\ref{fig:mean_sum_rates}. From Fig.~\ref{fig:mean_ISLs} we observe that, in all cases with $Q=2$, a slightly higher $\hat{\mu}_{M}$ is achieved with the GIEM algorithm than with the GEO algorithm, while the GMM algorithm achieves a slightly lower $\hat{\mu}_{M}$ for $P=8$. However, mixed results were obtained with $Q=1$, where the GEO algorithm leads to more established ISLs with $P\in\{6,8\}$. This performance degradation with even numbers of orbital planes was also observed for lower and greater values of $P$. The main reasons for this phenomenon are two-fold: 1) by not implementing the cross-seam ISLs, the number of orbital planes $P$ has a great influence in the connectivity graph $\set{G}$; and 2) the worst-case performance for greedy algorithms is closer to the lower bound of $w(\set{M}^\star)/2$ for even values of $P$ than for odd values. 
Note that, for the GEO algorithm, the maximum number of established inter-plane ISLs $Q(P-1)N_p/2$ is achieved only for $P=8$ even though full inter-plane ISL connectivity is guaranteed for $P\geq7$.

\begin{figure}[t]
    \centering 
    \subfloat[]{\includegraphics{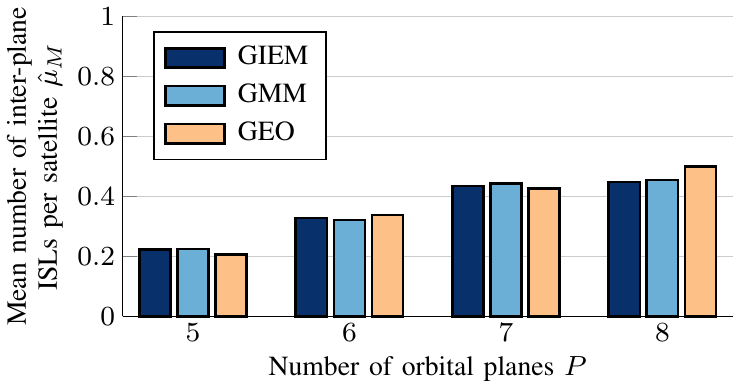}\label{fig:mean_ISLsq1}}\hfil
    \subfloat[]{\includegraphics{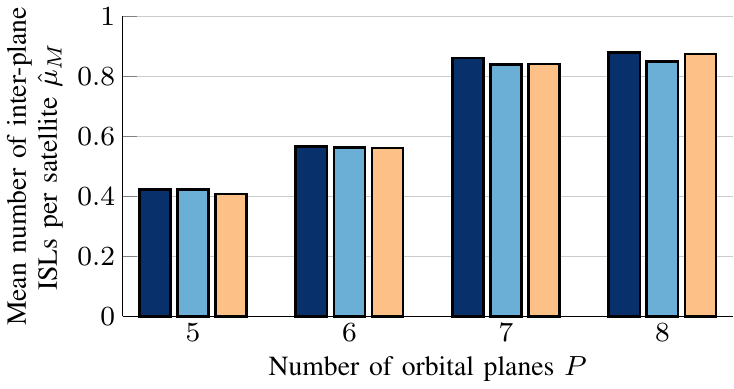}\label{fig:mean_ISLsq2}}
    \caption{Normalized mean number of established inter-plane ISLs per satellite $\hat{\mu}_{M}$ as a function of $P$ with (a) $Q=1$ and (b) $Q=2$.}
    \label{fig:mean_ISLs}
\end{figure}

\begin{figure}[t]
    \centering 
    \subfloat[]{\includegraphics{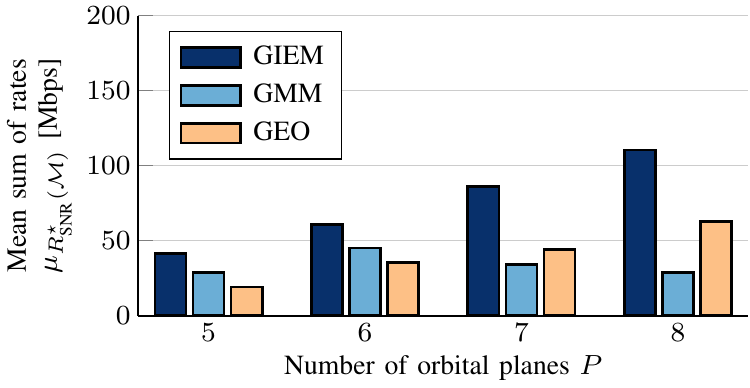}\label{fig:sum_rates_q1}}\hfil
    \subfloat[]{\includegraphics{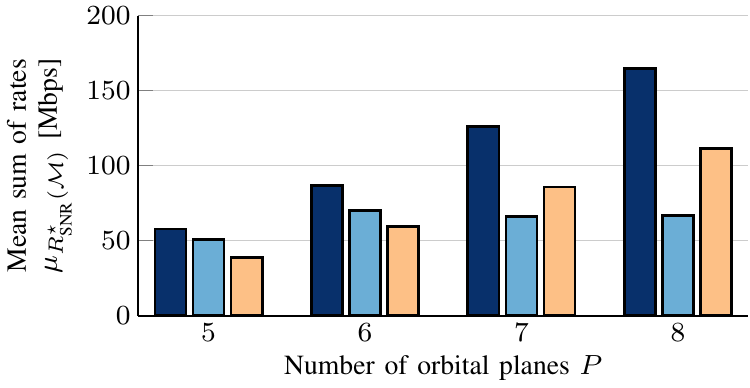}\label{fig:sum_ratesq2}}
    \caption{Mean sum of rates as a function of the SNR $\mu_{R^\star_\text{SNR}(\set{M})}$ with (a) $Q=1$ and (b) $Q=2$.}
    \label{fig:mean_sum_rates}
\end{figure}

Next, Fig.\ref{fig:mean_sum_rates} showcases the massive gains in the sum of rates provided by the GIEM algorithm when compared to the GMM and GEO algorithms. Specifically, even though the three matching algorithms establish a comparable number of ISLs (i.e., satellite pairs), the sum of rates with the GIEM algorithm is up to $115$\% higher than with the GEO algorithm for $P=5$ and up to $285$\% higher than with the GMM algorithm for $P=8$. The reason for this advantage w.r.t. the GEO algorithm is that the ISL connectivity is greatly limited to the polar regions with $P=5$ (see Fig.~\ref{fig:num_matches}), where the greedy algorithms select satellite pairs with greater data rates. On the other hand, the advantage of GIEM w.r.t. the GMM algorithm is because the latter maintains the ISLs for excessively long periods until the rates drop below $R_\text{min}$.

To finalize the performance evaluation of the satellite matching algorithms, Fig.~\ref{fig:CDFs} shows the empirical CDF of (a) the rates at each satellite $R^\star_\text{SNR}(uv)$ and (b) the propagation delay for $P=7$. As it can be seen, there is a great difference between the selected rates of the different satellites in the constellation. For example, with the GIEM algorithm, almost $50$\% of the rates are lower than $20$~kbps, less than $20$\% are above $100$~kbps, and only around $4$\% are higher than $1$~Mbps. Note that, due to its logarithmic scale, Fig.~\ref{fig:CDF_R_SNR} the great differences between the matching algorithms observed in Fig.~\ref{fig:mean_sum_rates} appear to be small. Besides, Fig.~\ref{fig:CDF_propagation_delay} shows that the propagation delay is less than $10$~ms in more than $80$\% of the established ISLs and only slight differences are observed between the matching algorithms. 

\begin{figure*}[t]
    \centering
    \subfloat[]{\includegraphics{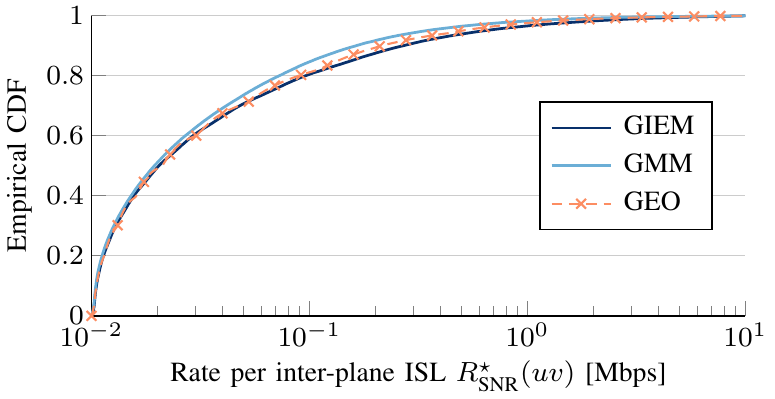}\label{fig:CDF_R_SNR}}\hfil
    \subfloat[]{\includegraphics{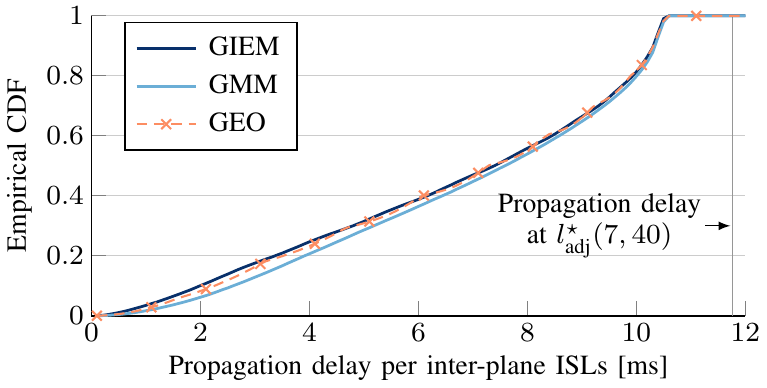}\label{fig:CDF_propagation_delay}}
    \caption{Empirical CDF of the (a) rates $R^\star_\text{SNR}(uv)$ and of the (b) propagation delay per inter-plane ISL for $P=7$. The maximum propagation delay at $l_\text{adj}^\star(7,40)=3527$~km is $11.77$~ms.} 
    \label{fig:CDFs}
\end{figure*}

Now we move on to assess the performance of the GRA algorithm. 
We are interested on finding the value of $K$ that maximizes the sum of rates with OFDMA and CDMA. for this, Fig.~\ref{fig:norm_mu_rates} shows the value of $\hat{\mu}_{R^\star_\text{SINR}(\set{A})}$ obtained with the GRA algorithm, round-robin and random resource allocation. The effective rates after resource allocation with OFDMA with $K$ resources were calculated by substituting $B$ with the sub-carrier bandwidth $B/K$ in~\eqref{eq:sinr} and~\eqref{eq:sinr_rate}. On the other hand, the effective data rates after resource allocation with CDMA with $K$ resources were calculated by dividing the rates calculated in~\eqref{eq:sinr_rate} by the spreading factor $1+\log_2(K)$. 

As it can be seen, the GRA algorithm clearly outperforms the two benchmark approaches. Specifically, $K=3$  and $K=2$ are optimal for the GRA algorithm with OFDMA and CDMA, respectively. With OFDMA, the maximum sum of rates with GRA algorithm is $28$\% greater than with round-robin allocation and $42$\% greater than with random allocation. On the other hand, with CDMA, the maximum sum of rates with GRA algorithm is $60$\% greater than with round-robin allocation and $71$\% greater than with random allocation. Also important to observe is that the sum of rates with OFMDA and its optimal the optimal $K$ is $84$\% greater than with CDMA and its optimal $K=2$.

\begin{figure}[t]
\centering
\includegraphics{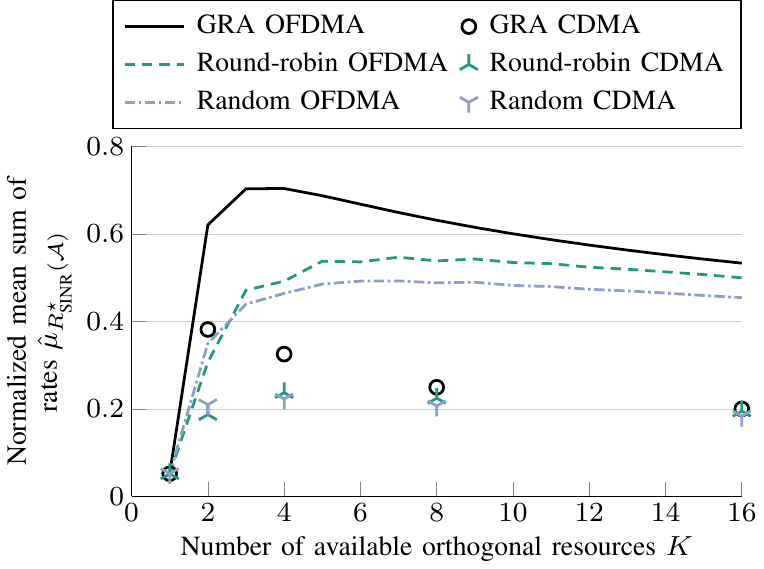}
\caption{Normalized mean sum of rates with OFDMA and CDMA with our GRA algorithm, along with round-robin and random allocation after satellite matching with the GIEM algorithm for $Q=2$ and $P=7$.}
\label{fig:norm_mu_rates}
\end{figure}
 
 
Finally, we compare the complexity of the satellite matching algorithms, along with that of the GIEM algorithm with GRA, which gives the complete solution to our problem in Fig.~\ref{fig:exec_time} for $Q=2$ and $K=P=7$. Clearly, with $\mathcal{O}\left(PN_p\right)$,as derived in Section~\ref{sec:isl_conn}, the GEO algorithm has the lowest complexity, followed by the GMM and the GIEM algorithms with $\mathcal{O}\left(PN_p^2(P-1)\right)$. Recall, that the difference in execution times between these two algorithms is mainly a result of sorting a shorter list with the GMM algorithm.

On the other hand, the complexity of our GRA algorithm is $\mathcal{O}\left(KM^2\right)$. We have observed that $M(n)\approx \mu_{M}\approx N_p(P-1)$ for all $n$, with $Q=2$ and $P=7$. Therefore, we have that the complexity of our GRA algorithm is $\mathcal{O}\left(KN_p^2(P-1)^2\right)$. This is much higher than the complexity of the satellite matching algorithms, performing in the order of $P-1$ times more operations than the GIEM algorithm. Hence, the complexity of establishing the ISLs is mainly determined by the resource allocation, as observed in Fig.~\ref{fig:exec_time}. 

\begin{figure}[t]
    \centering
    \includegraphics{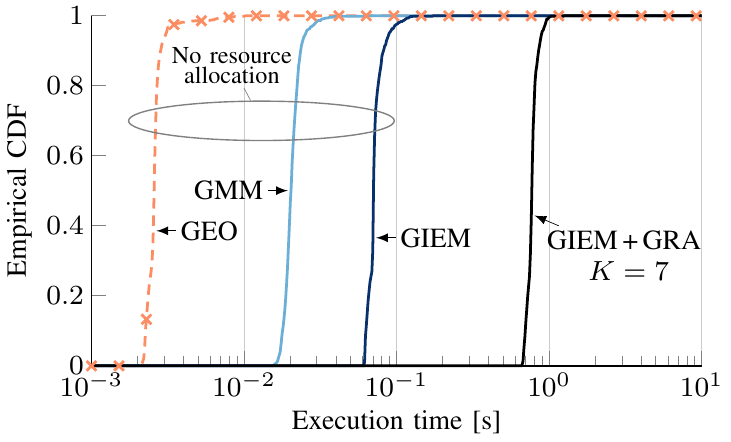}
        \caption{CDF of the execution time for the GIEM, GMM, and GEO, satellite matching algorithms, along with the GIEM algorithm with GRA given $Q=2$ and $K=P=7$.}
    \label{fig:exec_time}
\end{figure}

\section{Conclusions} \label{sec:conclusions} 
In this paper, we presented a framework to maximize the rates in the inter-plane ISLs of dense LEO constellations.
Furthermore, we provided a simple approach to constellation design, where we calculated the minimum transmission power and antenna gains to guarantee that all satellites have at least one potential inter-plane neighbor at all times (full inter-plane ISL connectivity).

Our results show that solving the satellite matching problem from scratch at each realization with our GIEM algorithm leads to a greater sum of rates, when compared to the GMM algorithm, where the previous satellite pairs are maintained for as long as possible, and to the benchmark GEO algorithm. In particular, the difference in performance between the GIEM and GMM algorithm grows with the density of the constellation. However, GMM algorithm reduces the execution time and the communication overhead due to frequent ISL establishment procedures w.r.t. the GIEM algorithm.

Regarding resource allocation, we observed that our algorithm provides massive gains when compared to random and round-robin resource allocation. Specifically, the maximum sum of rates achieved with our GRA algorithm is up to $71$\% greater than with random allocation. On the downside, our GRA algorithm has the greatest complexity of all the satellite matching and resource allocation algorithms considered in our analyses.

Therefore, in real LEO deployments, it is advisable to implement the GIEM and GRA algorithms together in highly dense LEO constellations whenever 1) the ISLs may be affected by interference, 2) the centralized control entity has sufficient computing power, and 3) the communication overhead of frequent ISL establishment is relatively low, as these algorithms will lead to the best performance among the considered algorithms. On the other hand, the GMM algorithm presents an efficient alternative if the communication overhead of establishing the ISLs is large and if no full inter-plane connectivity is guaranteed. 
\section*{Acknowledgment}

This work has been in part supported by the European Research Council (Horizon 2020 ERC Consolidator Grant Nr. 648382 WILLOW). 

\bibliographystyle{IEEEtran}
\bibliography{ISLassignment}

\end{document}